\documentclass[a4paper,UKenglish]{lipics}

\usepackage[utf8]{inputenc}
\usepackage{graphicx}
\usepackage{comment}
\usepackage{framed,color}
\usepackage{amsfonts}
\usepackage{tikz}
\usepackage[inline]{enumitem}
\usepackage{MnSymbol}
\usepackage{appendix}

\definecolor{shadecolor}{rgb}{0.9,0.9,0.9}
\definecolor{mygreen}{rgb}{0,0.5,0}
\definecolor{mygrey}{rgb}{0.5,0.5,0.5}

\usepackage{hyperref}

\newcommand{\arne}[1]{\textit{\textcolor{blue}{[arne]: #1}}} 
\newcommand{\N}{{\mathbb N}}
\newcommand{\tw}{\text{tree-width}}

\newcommand{\el}{\varepsilon}
\newcommand{\El}{\mathcal{E}}

\newcommand{\out}{\emph{out}}
\newcommand{\inc}{\emph{in}}
\newcommand{\RLF}{\textsf{RLF}}
\newcommand{\SLF}{\textsf{SLF}}

\usepackage{microtype}

\title{Transiently Consistent SDN Updates: Being Greedy is Hard}

\titlerunning{Greedy is Hard} 

\author[1]{Saeed Akhoondian Amiri}
\author[1]{Arne Ludwig}
\author[2]{Jan Marcinkowski}
\author[3,1]{Stefan Schmid}

\affil[1]{Technical University Berlin, 
  Berlin, Germany, 
  \texttt{saeed.amiri@tu-berlin.de, arne@inet.tu-berlin.de}}
\affil[2]{Institute of Computer Science,
University of Wroclaw, Wroclaw, Poland,
  \texttt{jasiekmarc@gmail.com}}
  \affil[3]{Aalborg University.
  Aalborg, Denmark,
  \texttt{schmiste@cs.aau.dk}}

\authorrunning{
S.~Amiri,
A.~Ludwig,
J.~Marcinkowski,
S.~Schmid
}

\begin{document}

\maketitle

\begin{abstract}
The software-defined networking paradigm introduces
interesting opportunities to operate
networks in a more flexible, optimized, yet formally
verifiable manner. Despite the logically
centralized control, however, a Software-Defined Network~(SDN)
is still a distributed system,
with inherent delays between the switches and the controller.
Especially the problem of changing network
configurations in a consistent manner,
also known as the consistent network update
problem, has received much attention 
over the last years. In particular, it has been shown
that there exists an inherent tradeoff between update
consistency and speed.
This paper revisits the problem
of updating an SDN in a transiently consistent,
loop-free
manner. First, we rigorously prove that computing a maximum
(``greedy'') loop-free network update is generally NP-hard;
this result has implications for the classic
maximum acyclic subgraph problem
(the dual feedback arc set problem) as well.
Second, we show that for
special problem instances, fast and good approximation
algorithms exist. 
\end{abstract}

\section{Introduction}\label{sec:intro}

By outsourcing and consolidating
the control over multiple data-plane elements
to a centralized software program,
Software-Defined Networks~(SDNs)
introduce flexibilities and optimization
opportunities. However, 
while a logically centralized
control is appealing, an SDN
still needs to be regarded as a distributed system,
posing non-trivial challenges~\cite{infocom15,correct,hotnets14update,roger,
sharon,abstractions,jukka}.
In particular, the communication channel between switches
and controller exhibits non-negligible and varying
delays~\cite{dionysus,abstractions},
which may introduce inconsistencies
during \emph{network updates}. 

Over the last years, the problem of how to
consistently update routes in a~(software-defined) network 
has received much attention, both in
the systems as well as in the 
theory community~\cite{correct,podc15,roger,abstractions,
brighten-nsdi}.
While in the seminal work by Reitblatt et
al.~\cite{abstractions}, 
protocols providing strong, per-packet consistency guarantees 
were presented~(using some kind of
2-phase commit approach),
it was later observed that weaker, but transiently consistent
guarantees can be implemented more efficiently.
In particular, Mahajan and Wattenhofer~\cite{roger} 
proposed a first algorithm to update routes in a network
in a transiently loop-free manner.
Their approach is appealing as it does
not require packet tagging~(which comes with overheads
in terms of header space and also introduces challenges
in the presence of middleboxes~\cite{simple} or multiple controllers~\cite{infocom15})
  or additional TCAM entries~\cite{infocom15,abstractions}
 ~(which is problematic given the fast table growth both in the Internet
  as well as in the highly virtualized datacenter~\cite{fib-growth}).
Moreover, this approach also 
allows~(parts of the) paths to become available sooner~\cite{roger}.

Concretely, to update a network in a transiently loop-free
manner, the approach proceeds \emph{in rounds}~\cite{podc15,roger}:
in each round, 
a ``safe subset'' of~(so-called OpenFlow) switches is updated,
such that, independently of the times and order in 
which the updates of this round take effect, 
the network is always consistent. 
The scheme can be implemented as follows:
After the switches of round~$t$
have confirmed the successful update~(e.g., using acknowledgments~\cite{update-ack}),
the next subset of switches for round~$t+1$ is
scheduled. 

It is easy to see that a simple update schedule always exists:
we can update switches one-by-one, proceeding from
the destination toward the source of a route.  
In practice, however, it is desirable that updates are fast
and new routes
become available quickly: 
Ideally, in order to be able to use as many new links
as possible, one aims to maximize the number of concurrently
updated switches~\cite{roger}. 
We will refer to this approach as the \emph{greedy approach}.

This paper revisits the problem of updating a maximum
number of switches in a transiently loop-free manner.
In particular, we consider the two different notions of loop-freedom
introduced in~\cite{podc15}: \emph{strong loop-freedom} and \emph{relaxed
loop-freedom}. The first variant guarantees loop-freedom
in a very strict, topological sense: no single packet will ever 
loop. The second variant 
is less strict, and allows for a small constant number of packets to loop
during the update; however, at no point in time should newly arriving packets 
be pushed into a loop. It is known that by relaxing loop-freedom, in principle 
many more switches can be updated simultaneously.

\noindent \textbf{Our Contributions.} 
We rigorously prove that computing the maximum set of 
switches which can be updated simultaneously, without introducing
a loop,
is NP-hard, both regarding strong and relaxed
loop-freedom. This result may be somewhat suprising, 
given the very simple
 graph induced by our network update problem.
The result also has implications for the classic
Maximum Acyclic Subgraph Problem~(MASP), a.k.a.~the 
dual Feedback Arc Set Problem~(dFASP): The problem of computing
a maximum set of switches which can be updated
simultaneously, corresponds to the dFASP, 
on special graphs essentially describing
two routes~(the old and the new one). 
Our NP-hardness result shows that MASP/dFASP
are hard even on such graphs.
On the positive side, we identify 
network update problems which allow for optimal or almost 
optimal~(with a provable approximation factor less than 2) polynomial-time algorithms,
e.g., problem instances where the number of leaves is bounded
or problem instances 
with bounded underlying undirected tree-width.

\section{Model}\label{sec:model}

We are given a network and two policies resp.~\emph{routes}
$\pi_1$~(the \emph{old policy}) and~$\pi_2$~(the \emph{new policy}).
Both~$\pi_1$ and~$\pi_2$ are simple directed paths (\emph{digraphs}).
Initially,
packets are forwarded~(using the \emph{old rules}, henceforth also called \emph{old edges}) along~$\pi_1$, and 
eventually they should be forwarded according to the new rules of~$\pi_2$. Packets
should never be delayed or dropped at a switch, henceforth 
also called \emph{node}: whenever a packet arrives at a node, a matching forwarding
rule should be present.
Without loss
of generality, we assume that~$\pi_1$ and~$\pi_2$ lead from a source~$s$ to a destination~$d$.

We assume that the network is managed
by a controller which sends 
out forwarding rule updates 
to the nodes.
As the individual node updates occur in an asynchronous manner, 
we require the controller to send out simultaneous updates only
to a ``safe'' subset of nodes. Only after these updates have been confirmed
(\emph{ack}ed), the next subset is updated.

We observe that nodes appearing only in one or none of the two paths are
trivially updatable, therefore we 
focus on the network~$G$ induced by the nodes~$V$ which are part of \emph{both}
policies~$\pi_1$ \emph{and}~$\pi_2$, i.e.,~$V = \{v: v \in \pi_1 \wedge v \in \pi_2\}$.
We can represent the policies as
$\pi_1=(s=v_1,v_2,\ldots,v_{\ell}=d)$ and~$\pi_2=(s=v_1,\pi(v_2),\ldots,\pi(v_{\ell-1}),v_\ell=d)$,
for some permutation~$\pi: V\setminus\{s,d\} \rightarrow V\setminus\{s,d\}$ and some number~$\ell$.
In fact, we can represent policies in an even more compact way:
we are actually only concerned about the
nodes~$U\subseteq V$ which need to be updated.
Let, for each node~$v\in V$,~$\out_1(v)$~(resp.~$\inc_1(v)$) denote the outgoing~(resp.~incoming) edge according to
policy~$\pi_1$,
and~$\out_2(v)$~(resp.~$\inc_2(v)$) denote the outgoing~(resp.~incoming) edge according to policy~$\pi_2$.
Moreover, let us extend these definitions for entire node sets~$S$, i.e.,~$\out_i(S)=\bigcup_{v\in S} \out_i(v)$, for
$i\in\{1,2\}$, and analogously, for~$\inc_i$.
We
define~$s$ to be the first node~(say, on~$\pi_1$) with~$\out_1(v) \neq \out_2(v)$, and~$d$
to be the last node with~$\inc_1(v) \neq \inc_2(v)$.
We are interested in the set of to-be-updated nodes~$U= \{v\in V: \out_1(v) \neq \out_2(v)\}$,
and define~$n=|U|$. Given this reduction, in the following,
we will assume that~$V$ only consists of interesting nodes~($U = V$).

We require that paths be loop-free~\cite{roger},
 and distinguish between \emph{Strong
Loop-Freedom}~($\SLF$)
and \emph{Relaxed Loop-Freedom}~($\RLF$)~\cite{podc15}.

\noindent \textbf{Strong Loop-Freedom.}
We want to find an \emph{update schedule}~$U_1, U_2, \ldots, U_k$, i.e., a sequence of subsets~$U_t\subseteq U$
where the subsets form a partition of~$U$~(i.e.,~$U=U_1 \cupdot U_2 \cupdot \ldots \cupdot U_k$),
with the property that for any round~$t$, given that the updates~$U_{t'}$ for~$t'<t$ have been made,
all updates~$U_t$ can be performed ``asynchronously'', that is, in an arbitrary order without violating loop-freedom.
Thus, consistent paths will be maintained for any subset of updated nodes,
independently of how long individual updates may take.

More formally, let~$U_{<t}=\bigcup_{i=1,\ldots,t-1} U_i$ denote the set of nodes which
have already been updated before round~$t$, and let~$U_{\leq t}$,~$U_{>t}$ etc.~be defined analogously.
Since updates during round~$t$ occur asynchronously, an arbitrary subset of nodes~$X \subseteq U_t$
may already have been updated while the nodes~$\overline{X}=U_t\setminus X$ still use the old rules,
resulting in a temporary forwarding graph~$G_t(U,X,E_t)$
over nodes~$U$, where~$E_t= \out_1(U_{>t}\cup \overline{X}) \cup \out_2(U_{<t}\cup X)$.
We require that the update schedule~$U_1, U_2, \ldots, U_k$
fulfills the property that for all~$t$ and for any~$X\subseteq U_t$,~$G_t(U,X,E_t)$
is loop-free.


In the following we will call an edge~$(u,v)$ of the new policy
$\pi_2$ \emph{forward}, if~$v$ is closer~(with respect to~$\pi_1$) to the destination,
resp.~\emph{backward}, if~$u$ is closer to the destination.
It is also
convenient to name nodes after their outgoing edges w.r.t. policy $\pi_2$~(e.g., \emph{forward} or \emph{backward});
similarly, it is sometimes convenient to say that we \emph{update an
  edge} when we update the corresponding node.

While the initial network configuration
consists of two paths, in later rounds,
the already updated solid edges may no longer form a line from left to right,
but rather an arbitrary directed tree, with tree edges directed towards the destination~$d$.
We will use the terms \emph{forward} and \emph{backward}
also in the context of the tree: they are defined with respect to the direction of the tree root.
However, there also emerges a third
kind of edges:
\emph{horizontal edges} in-between two different branches of the
tree.

%

\noindent \textbf{Relaxed Loop-Freedom.}
\emph{Relaxed Loop-Freedom}~($\RLF$) is motivated by the practical
observation that transient loops are not very harmful if they
do not occur between the source~$s$ and the destination~$d$.
If relaxed loop-freedom is preserved, only a constant number of packets
can loop: we will never push new packets into a loop ``at line rate''.
In other words, even if switches acknowledge new updates late~(or never), new packets
will not enter loops. Concretely, and similar to the definition of~$\SLF$, we require the update schedule
 to fulfill the property that for all rounds~$t$ and for any subset~$X$,
 the temporary forwarding graph~$G_t(U,X,E'_t)$ is loop-free.
The difference is that we only care about the subset~$E'_t$ of~$E_t$ consisting of edges \emph{reachable
from the source}~$s$.

\noindent \textbf{The Greedy Approach.}
Our objective is to update simultaneously
as many nodes~(or equivalently, edges) as possible: a
greedy approach~\cite{roger}. 
Note that in the first round, computing a maximum
update set is trivial: All forward edges can be updated
simultaneously, as they will never introduce a cycle;
at the same time, no backward edge can be updated
in the first round, as it can always induce a cycle.
Also observe that since all nodes lie on the path from
source to destination, this holds for both strong
and relaxed loop-freedom. However, as we will show in this paper,
already in the second round, a computationally
hard problem can arise.



\section{Being Greedy is Hard}\label{sec:hardness}

Interestingly, although the underlying graphs are very simple, 
and originate from just two (legal) paths, we 
now prove that the loop-free network update problem is NP-hard.

\begin{theorem}
\label{thm:greedyhard}
The greedy network update problem is NP-hard.
\end{theorem}

Our reduction is from the NP-hard 
\emph{Minimum Hitting Set}
problem. 
This proof is similar for both consistency models: strong
and relaxed loop-freedom, and we can present the
two variants together.
The inputs to the hitting set problem are:
\begin{enumerate}
  \item A universe of~$m$ elements~$\El=\{\el_1,\el_2,\ldots,\el_m\}$. 
\item A set 
$S=\{S_1, S_2, S_3, \ldots, S_k\}$ of~$k$ subsets~$S_i \subseteq \El$. 
\end{enumerate}

The objective is to find a subset~$\El'\subseteq \El$ of minimal size, 
such that each set~$S_i$ includes at least one element
from~$\El'$:~$\forall S_i \in S: S_i \cap \El' \neq \emptyset$. 
In the following, we will assume that elements are unique and 
can be ordered~$\el_1 < \el_2 \ldots < \el_m$.
The idea of the reduction is to create, in polynomial time, a legal network update
instance where the problem of choosing a maximum set of nodes which
can be updated concurrently is equivalent to 
choosing a minimum hitting set. 
While in the initial network configuration, essentially describing two paths
from~$s$ to~$d$, a maximum update set can be chosen in polynomial time
(simply update all forwarding edges but no backward edges),
we show in the following that already in the second round, 
the problem can be computationally hard. 

\begin{figure*} [ht]
    \begin{center}
      \includegraphics[width=0.79\textwidth]{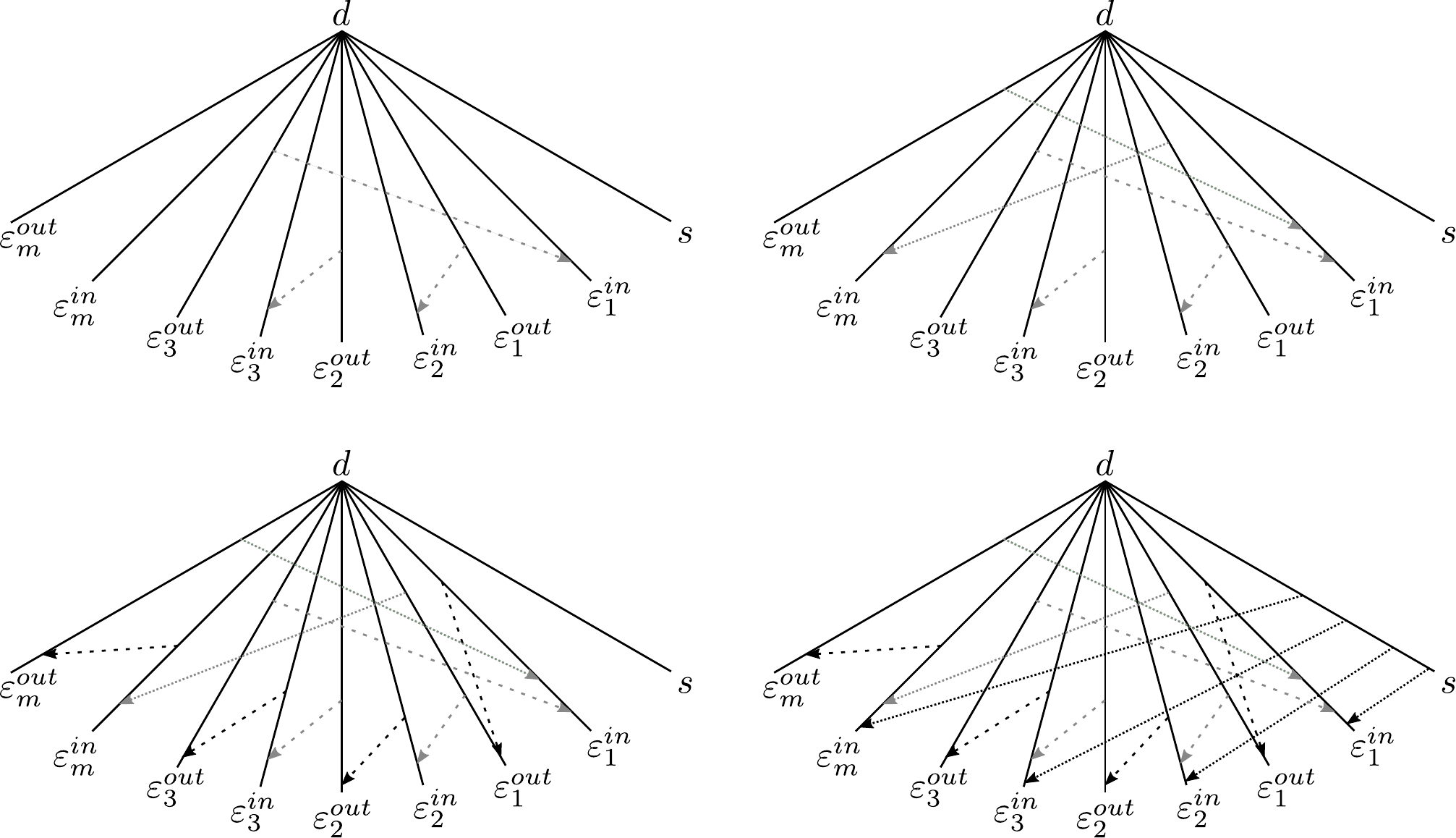}\\
      \caption{Example: Construction of network update instance
      given a hitting set instance with~$\El=\{1,2,3,\ldots,m\}$ and~$S=\{\{1,2,3\},\{1,m\}\}$. 
      Each element~$\el\in \El$ is 
represented by a pair of branches, one called outgoing~($out$) and 
one incoming~($in$). Moreover, we add a branch representing
the~$s-d$ path on the very right. The solid black branches represent already installed rules~(either old or updated in the first round), and new rules~(dashed) 
are situated between the branches.  
      There are three types of to-be-updated, dashed edges:
      one type represents the sets~(loosely and densely dashed grey), 
      one type represents element selector edges~(between in and out branch, loosely dashed black),
      and one type is required to connect the~$s-d$ path
      to the elements~(densely dashed grey).
      We prove that
      such a scenario can be reached after one
      update round where all~(and only)
      forward edges are updated. 
   \emph{Top-left:} Each loosely dashed grey edge represents~$m+1$ 
edges, and is used to describe the set~$\{1,2,3\}$:$(1,2),(2,3),(3,1)$. 
\emph{Top-right:} Each densely dashed grey edge represents~$m+1$ edges 
and is used for the set ~$\{1,m\}$:~$(1,m),(m,1)$. \emph{Bottom-left:} 
The loosely dashed black edges are single edges and are the element selector edges,
 representing the decision if
 an element is part of~$\El'$ or not. 
\emph{Bottom-right:} Each densely dashed edge 
visualizes~$m \cdot~(m+1)$ edges from the~$s$-branch to the incoming branches of every~$\el\in \El$.}
      \label{fig:overview-pic}
    \end{center}
  \end{figure*}

More concretely, based on a hitting set instance, we aim to construct a 
network update instance of the following form,
see Figure~\ref{fig:overview-pic}.
  For each element 
$\el \in \El$, we create a pair of branches~$\el^{in}$ and~$\el^{out}$,
i.e.,~$2m$ branches in total. 
To model the relaxed loop-free case, 
in addition to the~$\El$ branches, we add a source-destination branch, 
from~$s$ to~$d$, depicted on the right in the figure. 
We will introduce the following to-be-updated new edges:
\begin{enumerate}
\item \textbf{Set Edges~(SEs):} 
The first type of edges models sets.
Let us refer to the~(ordered)
elements in a given set~$S_i$ by 
$\el_1^{(i)}<\el_2^{(i)}<\el_3^{(i)} \ldots$. 
For each set~$S_i \in S$, 
we now 
create~$m+1$ edges from each~$\el_j^{(i)}$ to~$\el_{j+1}^{(i)}$,
in a modulo fashion. That is, we also introduce~$m+1$ edges from the last element 
to the first element of the set. These edges start at the~$out$ branch
of the smaller index and end at the  
$in$ branch of the larger index. There are no requirements
on how the edges of different sets are placed with respect to
each other, as long as they are not mixed. Moreover, only one instance of multiple equivalent SEs
arising in multiple sets must be kept. 

\item \textbf{Anti-selector Edges~(AEs):} 
These~$m$ edges constitute the decision problem of 
whether an element should be included 
in the minimum hitting set. 
AEs are created as follows: 
From the top of each~$in$ branch we create a \emph{single} edge
 to the 
bottom of the corresponding~$out$ branch. 
That is, we ensure that an 
update of the edge from~$\el_{i}^{in}$ to~$\el_{i}^{out}$ is equivalent to 
$\el_i \not\in \El'$, or, equivalently, 
every~$\el_i \in \El'$ will not be included in the update set. 

\item \textbf{Relaxed Edges~(WEs):} These edges are 
only needed for the relaxed loop-free case. They connect
the~$s$-$d$ branch to the other branches in such a way
that no loops are missed. In other words, the edges aim to emulate
a strong loop-free scenario by introducing artificial sources
at the bottom of each branch.
To achieve this, we create a certain number of
edges from the~$s$-branch to the bottom of every~$in$ branch. The precise amount will be explained at the detailed 
construction part of creating parallel edges.
See Figure~\ref{fig:overview-pic}~\emph{bottom-left} for an example. 
\end{enumerate}

The rational is as follows. 
If no \emph{Anti-selector Edges}~(AEs)
are updated, 
all \emph{Relaxed Edges}~(WEs)
as well as all \emph{Set Edges}~(SEs)
can be updated simultaneously, without
introducing a loop. 
However, since there are in total exactly~$m$ 
AEs but each set of SEs
are~$m+1$ edges~(hence they will all be updated), we can conclude that
the problem boils down to selecting a
maximum number of element AEs
which do not introduce a loop. 
The set of non-updated AEs
constitutes the selected sets, the hitting set:
There must be at least one
element for which there is an AE,
preventing the loop. By maximizing the number
of chosen AEs~(maximum update set)
we minimize the hitting set. 

Let us consider an example: 
In Figure~\ref{fig:overview-pic}~\emph{bottom-right}, 
if for a set 
$S_i$ every AE of~$\el_i \in S_i$ is updated, a cycle is created:
updating edges~$\el_{1}^{in}$ and~$\el_{m}^{in}$ 
results in a cycle with the~$m+1$ 
edges from~$\el_{1}^{out}$ and~$\el_{m}^{out}$.
Note that the resulting network update
instance is of polynomial size~(and can also
be derived in polynomial time). 
In the remainder of the proof,
we show that 
the described network update instance is indeed legal,
e.g., we have a single path from source to destination,
and this instance can actually be obtained after one update round.

\subsection{Concepts and Gadgets}\label{sec:Concepts}

Before we describe the details of the construction, 
we first make some fundamental observations
regarding greedy 
updates.


\textbf{Introducing Forwarding Edges and Branches:} 
First, a delayer concept is required to establish
forwarding edges for the second round. 
Observe that every forwarding edge~$(a,b)$, with~$a<b$,
is always updated by a greedy algorithm in the first round. 
A delayer is used to construct a forward edge~$(a,b)$, with~$a<b$, that 
is created in the second round.
A \emph{delayer} for edge~$(a,b)$ consists of two edges: an 
edge pointing backwards to~$a'$ from~$a$
with~$a' < a$, plus an edge pointing from there to~$b$. 
The forward edge~$(a',b)$ will be updated in the first round, 
which yields an edge~$(a,b)$ due to merging~(see
Figure~\ref{fig:delayer-pic}). 

\begin{figure} [ht]
    \begin{center}
      \includegraphics[width=0.45\columnwidth]{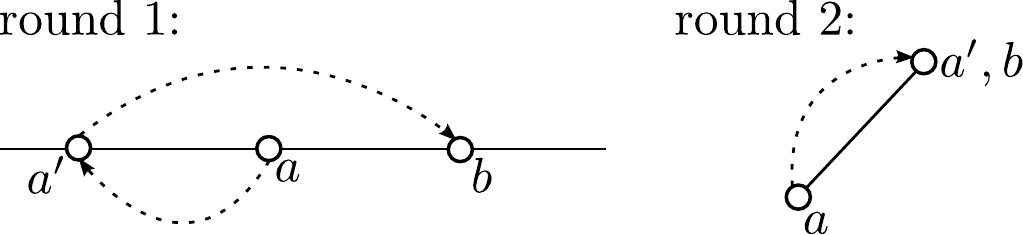}\\
      \caption{Delayer concept: A forwarding edge~$(a, a'b)$ 
      can be created in round 2 using a helper node~$a'$.}
      \label{fig:delayer-pic}
    \end{center}
  \end{figure}

We next describe how to 
create the~$in$ and~$out$ branches as well as the~$s$ branch pointing to the destination
$d$
(recall Figure~\ref{fig:overview-pic}). This can 
be achieved as follows: From a node close to the source~$s$, 
we create a path of forward edges which ends at the 
destination. Each of these forward edges will be updated in the first round, 
and hence merged with its respective 
successor, which will be the destination for the very last forward edge. The nodes belonging to these forward edges 
will be called \emph{branching nodes}. 
Every node in-between two \emph{branching nodes} will be part of a new branch pointing to the 
destination. See Figure~\ref{fig:branching-pic} for an example. The rightmost node before the \emph{branching node} on 
the 
line will also be the topmost node on the branch after the first round update~(as long as it has an outgoing backward 
edge, hence not being updated in the first round). 
We will use the terms right and high 
(rightmost-topmost) and left-low for the first and second round interchangeably.

\begin{figure} [ht]
    \begin{center}
      \includegraphics[width=0.69\columnwidth]{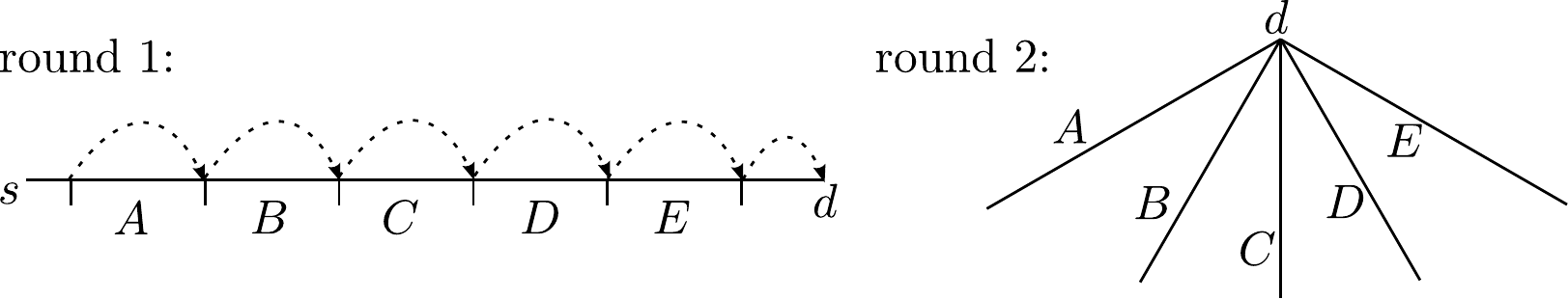}\\
      \caption{Creating branches after a greedy update of forward edges.}
      \label{fig:branching-pic}
    \end{center}
  \end{figure}

\textbf{Introducing Special Segments:}
In our construction, we split the line~(old path) into disjoint 
segments which will become independent branches at the 
beginning of the second round. 
In addition to these segments, there will be two special segments,
one at the beginning and one at the end. 
The first will not even become an independent branch 
at the beginning of the second round, but is merely
used to realize the delayer edges.
Behind the very last segment~($\el_{1}^{in}$) and just before~$d$,
there is a second special segment, which
we call \emph{relaxed}: it is needed to create the branch 
with the source~$s$ at the bottom and its connections to the 
other~$\el_{i}^{in}$ branches. 

In our construction, SEs come in groups 
of~$m+1$ edges. These edges must eventually
be part of a legal network update path, and
must be connected in a loop-free manner.
In other words, 
to create the desired problem instance, 
we need to find a way to connect two branches~$b_1$ and~$b_2$ with~$m+1$ edges, 
such that there is a single complete path from~$s$ to~$d$. 
Furthermore, these edges should not form a loop. 

\textbf{Creating Parallel Edges:}
Parallel
edges can be constructed as follows, henceforth called the \emph{zigzag-approach}
(Figure~\ref{fig:nConnection-pic}): Split the branch~$b_2$ into two different parts. The first part~$b_{2-b}$ on the 
left side~(respectively bottom of the branch)
will be used to complete the path but can only be reached over backward edges. The second part~$b_{2-t}$ will receive 
the 
incoming edges from the other branch,~$b_1$. Start at a node, say~$v_{o-1}$ on~$b_1$. Here create an edge 
to a 
node of~$b_{2-t}$, say~$v_{i-1}$ and from there a backward edge to a node of~$b_{2-b}$, say~$v'_{i-1}$. 
Afterwards use a delayed edge to connect to~$v_{o-1}$'s right~(respective to the line) neighbor, 
$v_{o-2}$. From here create the next edge to~$v_{i-1}$'s right neighbor,~$v_{i-2}$ and the backward edge to~$v'_{i-2}$ 
on~$b_{2-b}$ again. Repeat this procedure~$m+1$ times. 
  
This zigzag construction indeed ensures loop-freedom. To see this, note that all incoming edges from the~$b_{1}$ 
branch will always connect to the~$b_{2-t}$ part of~$b_2$. From here the way back to~$b_1$~(or potentially any 
other branch that connects with~$b_{2-t}$) can only be completed if any of the backward edges from~$b_{2-t}$ to 
$b_{2-b}$ has 
been updated. This cannot be true
for the strong loop freedom definition, since no backward edge can ever be 
updated and the edge is backward in the first and the second round. For relaxed loop freedom it also cannot be updated 
in the first round since it would create a loop on the~$s-d$ path, which is a line of all nodes in the first round. In 
the second round it will not be included since we make 
sure that a maximum update always includes the WEs which will be incoming at the very left side of 
$b_{2-t}$, and hence cannot be updated in the same round with any backward edge on this branch. 

In order to ensure that all the WEs 
will always be included, we will create~$m\cdot~(m+1)$ WEs to every \emph{in} branch. This is always more than the 
amount of backward edges on a 
single branch~$b_2$ since they are only created as a path completion for the SEs. We will have at most~$(m-1) \cdot 
(m+1)$ SEs incoming in a case where this node is connected to every other node~(but itself). Choosing the WEs will 
immediately force that none of the backward edges from~$b_{2-t}$ to~$b_{2-b}$ will be included, as they might cause a 
cycle on a path that might be in-between~$s$ and~$d$.

The~$m \cdot~(m+1)$ WEs to a branch~$b$ are simple
to create. Here, we do not need to take care about other branches 
reached from the \emph{relaxed} branch. Hence we can create the way back to the \emph{relaxed} branch without the 
detour over the~$b_{t}$ part. This is because the WEs will always be the incoming edges on the leftmost part of~$b_t$ 
without the possibility of any other parallel edges making use of them.
\begin{figure} [ht]
    \begin{center}
      \includegraphics[width=0.65\columnwidth]{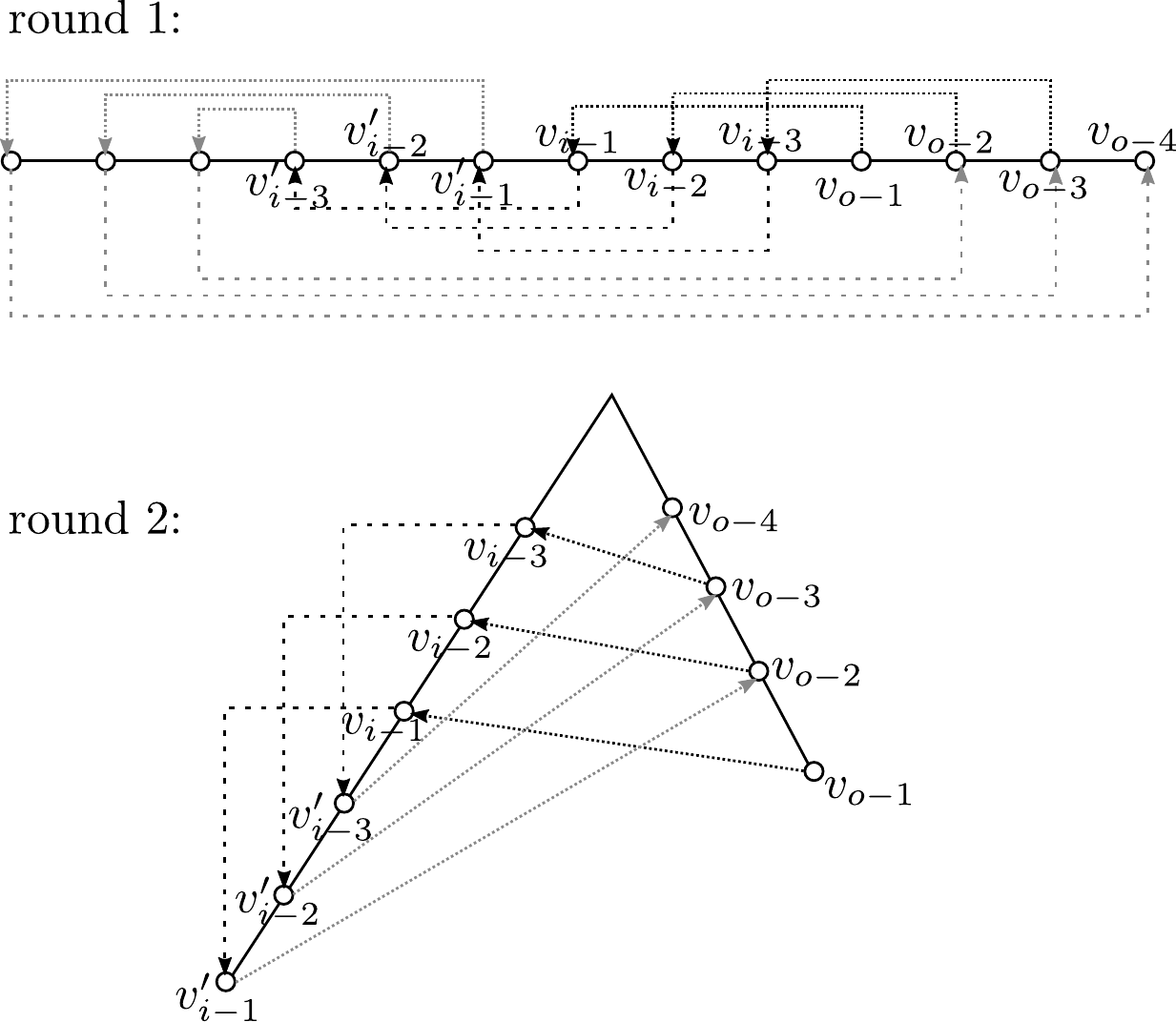}\\
      \caption{Connecting two branches with 3 edges. The backward edges shown in loosely dashed black assure that there will not be a 
way back in the second round from the \emph{in} branch to the \emph{out} branch.}
      \label{fig:nConnection-pic}
    \end{center}
  \end{figure}
\hfill\qed

 \subsection{Connecting the Pieces}
 
Given these gadgets, we are 
able to complete the construction of our problem instance.

\textbf{Realizing the Delayer:} The first created 
segment, \emph{temp}, serves
for edges that are created using 
the \emph{delayer} concept. This is due to 
our construction: every node that will be created in this interval
in our construction will be a forward node and therefore updated in the first 
greedy round. The \emph{temp} segment will be located right 
after the source~$s$ on the line. 

\textbf{Realizing the Branches: } We create two segments for 
each~$\el \in \El$, one~$out$ and one~$in$, and 
sort them
 in descending global order~(and depict them from left to right) 
 w.r.t.~$\el \in \El$, with the~$out$ segment closer to~$s$
 than the~$in$ segment for each~$\el$, i.e. 
$\el_{m}^{out}, \el_{m}^{in}, \ldots, \el_{2}^{out},$~$\el_{2}^{in},\el_{1}^{out}, \el_{1}^{in}$.

 \begin{figure} [ht]
    \begin{center}
      \includegraphics[width=0.75\columnwidth]{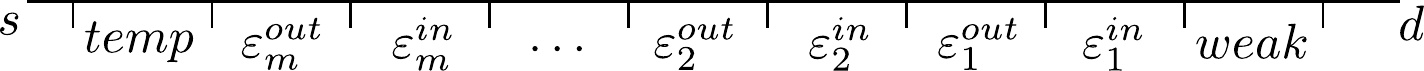}\\
      \caption{Illustration of how to split 
      the old line into segments according to the amount 
      of needed branches in the second round.}
      \label{fig:splitting-pic}
    \end{center}
  \end{figure}
 
\textbf{Connecting the Path:} We will now create the new path from the source~$s$ to the
destination~$d$ through all the different segments. 
This path requires additional edges. We will ensure that
these edges can always be updated and hence 
do not violate the selector properties. Moreover,
we ensure that they do not introduce a loop.
In
order to create a branch with~$s$ at the bottom 
(to ensure that the proof will also hold for relaxed loop-freedom), 
we start 
our path from the source~$s$ to a node~$relaxed-bot$ on the very left part of 
the \emph{relaxed} segment. From here we need to create the~$m \cdot~(m+1)$ connections to every other 
$\el_{i}^{in}$ 
branch, more precisely to the very left of the top part of this branch~$\el_{i-t}^{in}$: the relaxed 
Edges~(WEs). Starting from~$relaxed-bot$, we create the~$m \cdot~(m+1)$ zigzag edges we
postulated earlier~(see Section~\ref{sec:Concepts}) to the~$\el_{1}^{in}$ 
segment. Once this is done, we repeat this process for the remaining 
$\el_{i}^{in}$ connecting them in the same order 
blockwise, as they are ordered on the line. See Figure~\ref{fig:relaxedBranch-pic}.
 
 \begin{figure} [ht]
    \begin{center}
      \includegraphics[width=0.65\columnwidth]{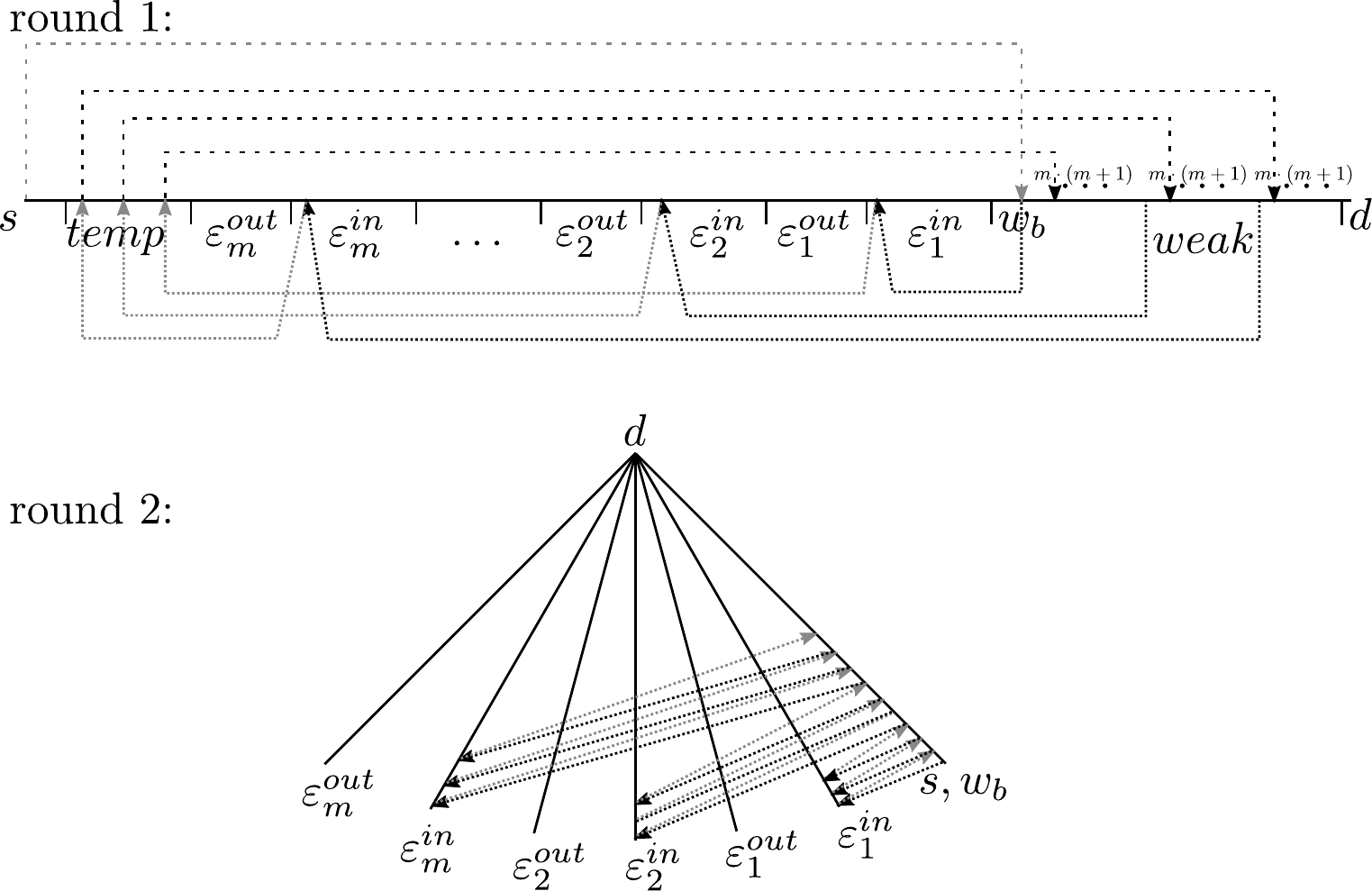}\\
      \caption{Creating the branch with the source at the bottom 
      and~$m \cdot~(m+1)$ connections to each~$\el_{i}^{in}$ segment of the 
line, as shown in Section~\ref{sec:Concepts}. The~$m \cdot~(m+1)$ connections are visualized as a single edge in the 
first round to enhance visibility.}
      \label{fig:relaxedBranch-pic}
    \end{center}
  \end{figure}
 
At the beginning of the second round, 
we will now have a branch with the source~$s$ at the bottom and 
$m+1$ edges to each of 
the~$\el_{i}^{in}$ branches. The next step is to connect the 
out branches with the in branches~(the Set Edges). 
For each set~$S_j \in S$ 
and each pair~$\el_i, \el_l \in S_j$ with no~$\el'\in S_j, \el_i < \el' < \el_l$, 
we create~$m+1$ edges from~$\el_{i}^{out}$ to 
$\el_{l}^{in}$, more precisely to the top part~$\el_{l-t}^{in}$ somewhere above the WEs. Each pair 
$\el_i, \el_l$ only needs to 
connect once with the~$m+1$ 
edges, even if it occurs in several 
different sets of~$S$. The last element~$\el_i$ of a set~$S_j$ will additionally need to be connected to the 
first 
element of the set~(the modulo edges).
 
After the~$m+1$ connections to~$\el_{m}^{in}$, the path returns at 
the right most~(or highest in the~$(s,d)$-branch)
node 
in the \emph{relaxed} segment. From here we 
create a backward edge to the left part of~$\el_{1}^{out}$. Here, we create~$m+1$ connections to every 
$\el_{i}^{in}$, which is 
the next larger element in any of the sets. An example is shown in Figure~\ref{fig:connectBranches-pic}.
 \begin{figure} [ht]
    \begin{center}
      \includegraphics[width=0.65\columnwidth]{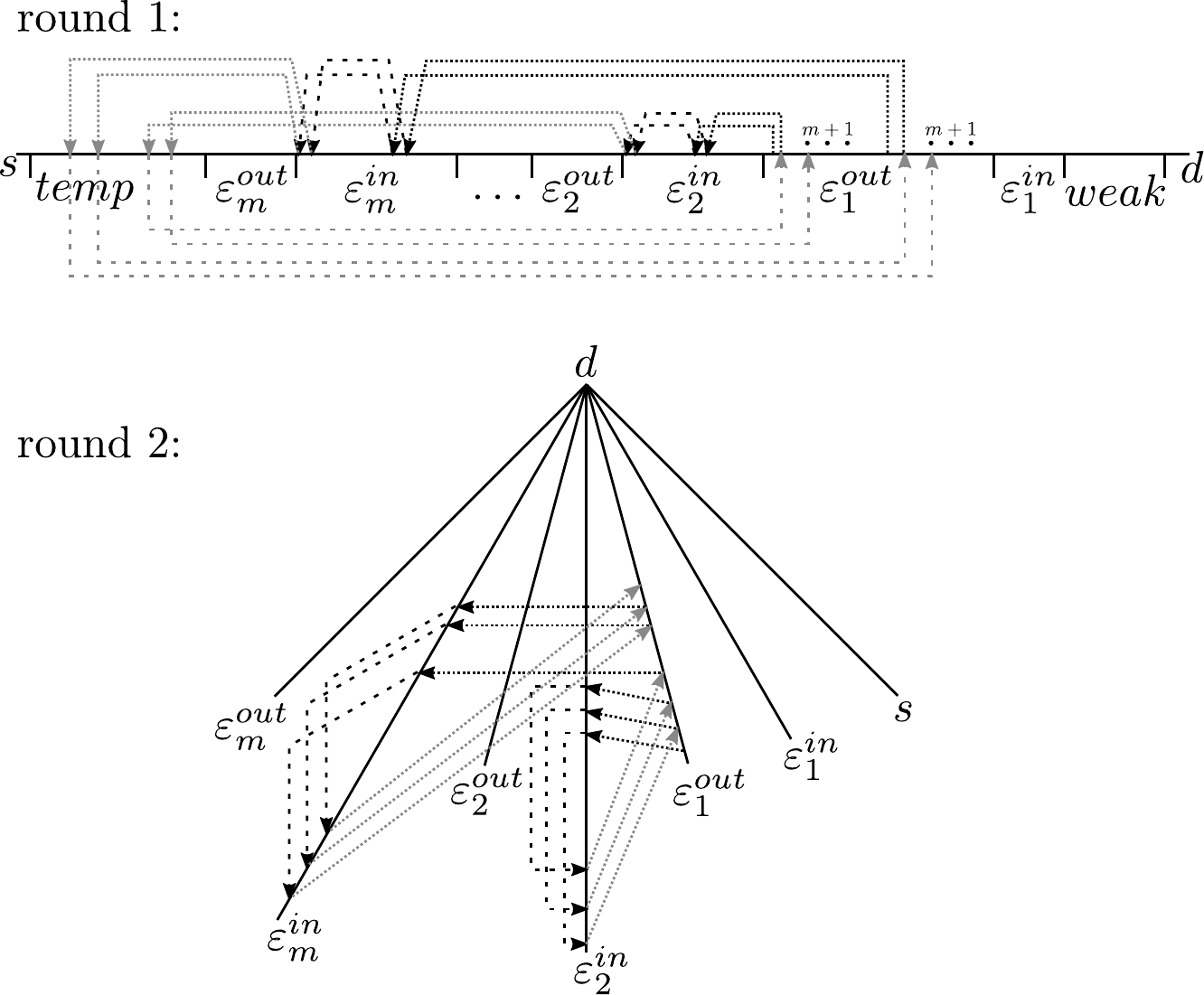}\\
      \caption{Connecting the~$\el_{1}^{out}$ branch with the branches~$\el_{2}^{in}$, 
$\el_{3}^{in}$,~$\el_{m}^{in}$. This scenario 
would be created for the sets:~$\{1,2,\ldots\}$,$\{1,3,\ldots\}$,~$\{1,m\}$. The densely dashed black edges show the outgoing 
edges from~$\el_{1}^{out}$. The loosely dashed black edges are the backward edges from the top part of a branch 
$\el_{i}^{in}$ to its bottom part~($\el_{i-t}^{in}$ to~$\el_{i-b}^{in}$). The densely dashed grey edges are the way back from 
$\el_{i}^{in}$ to~$\el_{1}^{out}$ and are needed to complete the path.}
      \label{fig:connectBranches-pic}
    \end{center}
  \end{figure}
 
To complete the~$m+1$ connections for every pair, we proceed as follows:
we connect the~$\el_{1}^{out}$ branch to all required in-branches, then add the edge from 
$\el_{1}^{out}$
to the~$\el_{2}^{out}$ branch, then add the edges from the~$\el_{2}^{out}$ branch 
to all required in-branches, etc. Generally, we interleave adding the 
edges from the~$\el_{i}^{out}$
branch to all required in-branches and then add the~$i$-out to~$(i+1)$-out
edge.
Until the path 
arrives at the end of the last out branch,~$\el_{m}^{out}$:
 \begin{itemize}
  \item \emph{Step A - Create  the ~$m+1$ set specific edges:} Here we create 
 ~$m+1$ connections to every successor in the 
  respective sets~(at most once per pair). If 
this element is the largest element in a set, it needs to be connected to the in part of the smallest element of this 
set again. Here the delayer concept needs to be used for the modulo edges.

  \item \emph{Step B - Connecting the out branches:} In order to create the next~$m+1$ connections 
  from the next out segment~$\el_{i+1}^{out}$, we need to connect it from our current out segment 
$\el_{i}^{out}$. The edge 
therefore needs to 
point to the rightmost part of~$\el_{i+1}^{out}$. Since this edge is always a 
backward edge in the first round~(we start closer to the destination and move backward towards the source), 
it will turn out to be an edge which points to the very 
top of~$\el_{i+1}^{out}$ at the beginning of the second round. This assures that there are no loops created, 
since the 
only way is going directly towards the destination. 
From here we create an edge pointing to the 
very left side of~$\el_{i+1}^{out}$~(evolving to a backward rule from top to bottom of the branch in the second 
round, 
hence not being part of the update set in the first nor the second round).

 \end{itemize}
 
  \begin{figure} [ht]
    \begin{center}
      \includegraphics[width=0.65\columnwidth]{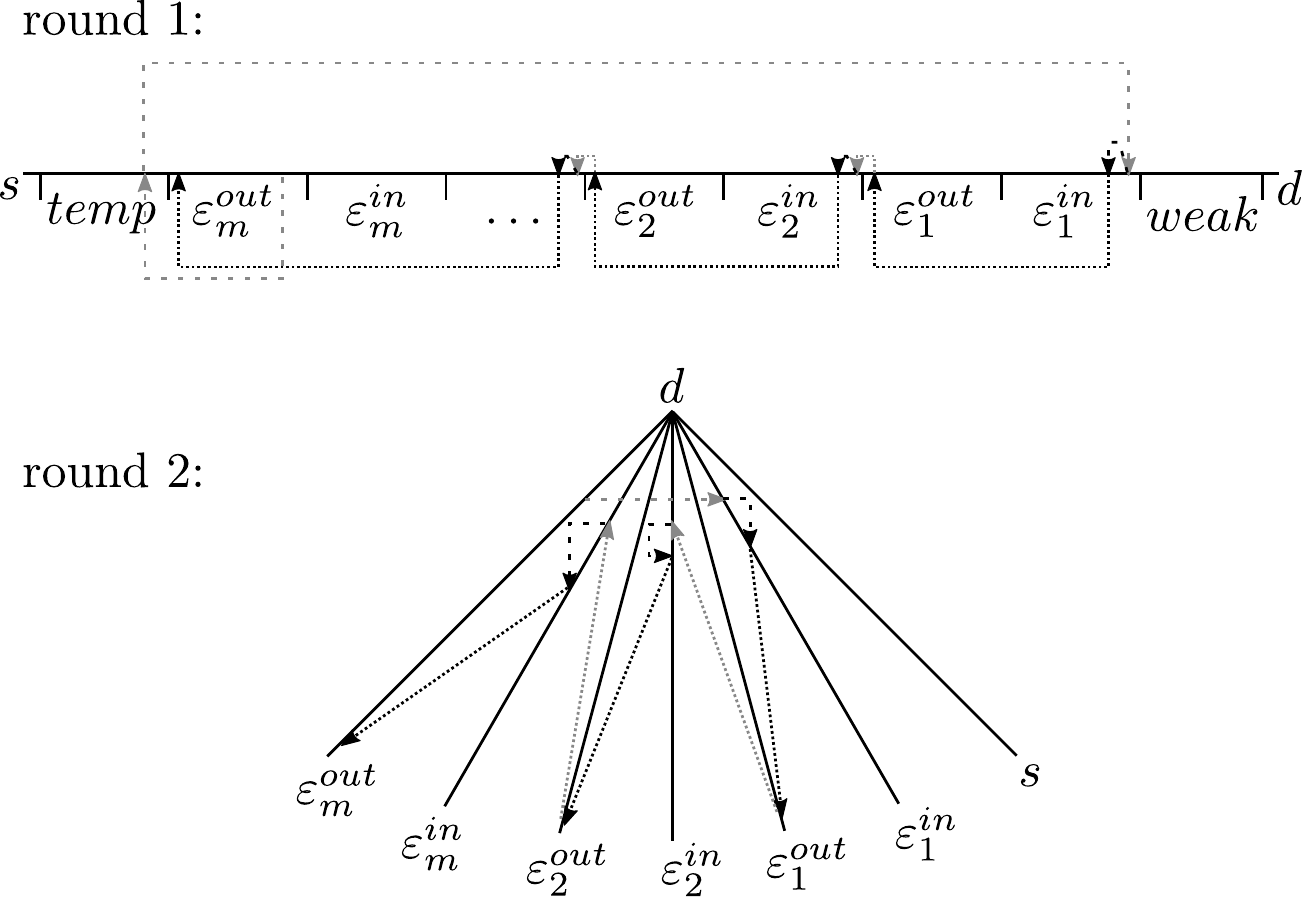}\\
      \caption{Connecting the in and out branches of every~$\el_i$, 
      shown in densely dashed black. The edges shown in desnely dashed grey are needed to 
keep the path complete and the backward edges in loosely dashed black are needed to ensure that only the destination can be reached 
from that point in the second round.}
      \label{fig:hittingEdges-pic}
    \end{center}
  \end{figure}
 
To finish the construction, we need to add the anti-selector edges~(AEs),
and connect the in and out branches 
of every single~$\el_i$ with each other.
The goal is to create, for each given~$i$, an edge from the top of each~$\el_{i}^{in}$ 
to the bottom of each~$\el_{i}^{out}$. This way, if this 
edge is included in the update, a loop may be formed: 
as every incoming edge to~$\el_{i}^{in}$ 
arrives below the AEs start 
point and every outgoing edge on~$\el_{i}^{out}$ is above AE's destination. 
The decision to not include one of these 
edges is equivalent to~$\el_i \in \mathcal{E}'$ in the minimum hitting set problem. 
In order to keep the path connected we will 
also need to include edges from~$\el_{i}^{out}$ to~$\el_{i+1}^{in}$, compare 
Figure~\ref{fig:hittingEdges-pic}. These edges 
will point to the top of~$\el_{i+1}^{in}$ and therefore do not create loops, since the only way is going 
directly to the 
destination. From here we create another backward edge to 
its left neighbor such that there is no possible other way than traversing towards~$d$ from this point. Without this 
backward edge loops may be created, since it introduces 
connections between branches which are not 
both in a set~$S_i$ of the hitting set problem.
Therefore, an update of one of the additional connector edges will never lead to a loop,
 and the edges can all be 
included in the update set of the round 2.

The construction of these edges is straightforward. From 
the end of the current path which is located on the~$\el_{m}^{out}$ segment, 
we create a delayed edge~(over \emph{temp}) to 
the very right part of the~$\el_{1}^{in}$ segment. From here we construct the path as described with a short 
backward edge 
to its left neighbor and then to the very left part of the 
$\el_{i}^{out}$ segment and again to the very right part of the~$\el_{i+1}^{in}$ segment afterwards, 
until we arrive at the very left part of the~$\el_{m}^{out}$ segment.
 
It remains to create the segments and branches for the second round. From~$\el_{m}^{out}$, 
 we create a backward edge 
to the \emph{temp} part. From here we use the branching concept and connect all horizontal nodes 
in-between the single parts that we created on the line~(see Figure~\ref{fig:finalStep-pic}).
 
 \begin{figure} [ht]
    \begin{center}
      \includegraphics[width=0.65\columnwidth]{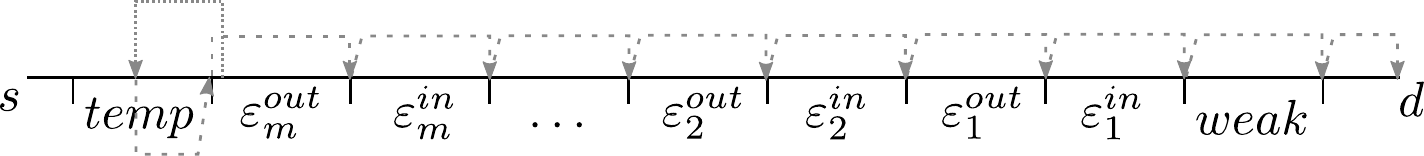}\\
      \caption{Connecting the segments with forward edges. This creates a single branch from the destination for every segment 
due to the merging. The edge shown in loosely dashed grey is connecting this step with the step before.}
      \label{fig:finalStep-pic}
    \end{center}
  \end{figure}

In summary,
 we ensured that already after a single greedy first update round, we end up 
in a situation where choosing the maximum set of updateable nodes is equivalent to choosing the minimum hitting set.

\section{Polynomial-Time Algorithms}\label{sec:algos}

While the computational hardness is disappointing,
we can show that there exist several interesting specialized
and approximative algorithms.

\noindent \textbf{Optimal Algorithms.}
There are settings where an optimal solution can be computed
quickly. For instance, it is easy to see that in the first round,
in a configuration with two paths, updating all forward edges
is optimal: Forward edges never introduce any loop, and at the same
time we know that backward edges can never be updated in the first round,
as any backward edge edge alone~(i.e., taking effect in the first
round), will immediately introduce a loop.
 In the following, we first present an optimal algorithm for~$\SLF$,
for trees with only two leaves.
We will then extend this algorithm to~$\RLF$.

\begin{lemma}
\label{lem:strongLF}
A maximum~$\SLF$ update set can be computed in polynomial-time
in trees with two leaves.
\end{lemma}
\begin{proof}
Recall that there are three types of new edges in the graph~(see also Figure~\ref{fig:edges}):
forward edges~($F$), backward edges~($B$) and horizontal edges~($H$), 
hence~$E=H\cup B \cup F$. Moreover, recall that 
forward edges can always be updated while backward
edges can never be updated in~$\SLF$. 
Thus, the problem boils down to selecting a maximum subset
of~$H$, pointing from one
branch to the other. If there is a simple loop
$C \in G$ such that~$H^C = E(C)\cap H \neq \emptyset$,
then~$|H^C|=2$ and we say that the
two edges~$e_1,e_2 \in H^C$ cross each other, 
written
$e_1 \times e_2$. 

We observe that the different edge types can
be computed efficiently. For illustration, suppose the policy graph~$G=(V,E)$~(the union 
of old and new policy edges)
is given as a
straight line
drawing~$\Pi$ in the 2-dimensional Euclidean plane, 
such that the old edges of the 2-branch tree
form two disjoint segments which meet at the root 
of the tree~(the destination), and such that each node
is mapped to a unique location.
Given the graph, such a drawing~(including crossings)
in the plane can be computed
efficiently.
Also note that there could be other edges which intersect w.r.t.~the drawing~$\Pi$,
 but those are not important for us.
 
Now create an auxiliary graph~$G'=(V',E')$ where~$V'=\{v_e\mid e\in H\}$, 
$E'=\{(v_{e_1},v_{e_2})\mid e_1,e_2 \in H : e_1 \times e_2\}$. The
graph~$G'$ is bipartite, and therefore finding a minimum vertex cover
$VC \in V(G)$ is equivalent to 
finding maximum matching, which can be done
in polynomial time. Let~$H' = \{e\mid e \in H : v_e\in VC\}$, then the set~$H'$ is a minimum
size subset of~$H$ which is not updatable. Therefore the set~$H\setminus H'$ is the
maximum size subset of~$H$ which we can update in a~$\SLF$ manner.

We conclude the proof by observing that all these algorithmic steps can be computed
in polynomial time.
\hfill\qed
    \begin{figure} [ht]
    \begin{center}
      \includegraphics[width=0.4\columnwidth]{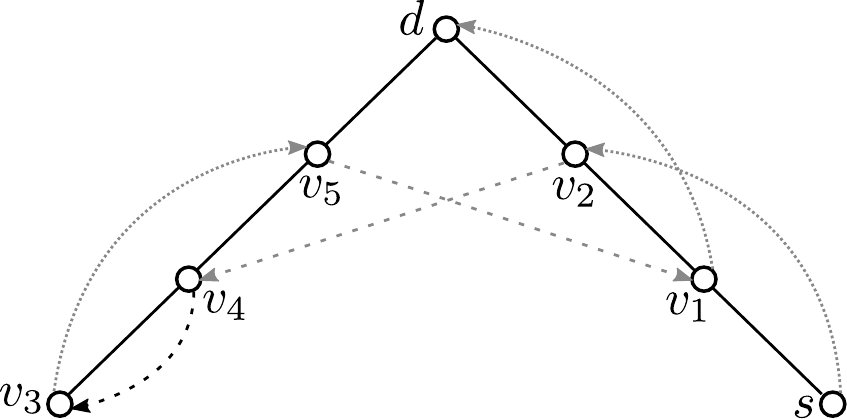}\\
      \caption{Concept of horizontal edges shown in loosely dashed grey. Both horizontal edges~$(v_2, v_4)$ and~$(v_5, v_1)$ are 
crossing each other. The backward edge~$(v_4, v_3)$ is shown in loosely dashed black and the forward edges in densely dashed grey. Note that $s$ does not necessarily have to be
a leaf.}
      \label{fig:edges}
    \end{center}
  \end{figure}
\end{proof}

\begin{lemma}
\label{lemma:weakLF}
A maximum~$\RLF$ update set can be computed in polynomial-time
in trees with two leaves.
\end{lemma}
\begin{proof}
  We prove the lemma by presenting a polynomial-time
reduction to the strong loop-free
case. Let us fix the path~(i.e., branch) in the tree 
consisting of the currently active edges which includes
both the source 
and the
destination:~$P^s_d=(s=v_0,\ldots,v_n=d)$.
Note that in the branch which
contains~$s,d$ there may exist some vertices which have a path to
$s$: those vertices are irrelevant for our construction and we just
consider the path~$P^s_d~$ of the old policy starting at~$s$.

Let us refer to the entire path in the other branch by
$P_2=(u_1,\ldots,u_m)$, omitting the vertex~$d$. 
Here, node 
$u_1$ is the node with the lowest~$y$-coordinate in
the drawing~$\Pi$ (for definition of $\Pi$ see the proof of Lemma~\ref{lem:strongLF}). In this case, we can update~$B$ edges
as long as they are not in any path from~$s$ to~$d$. Therefore, the objective
is to find the maximum subset~$S \subseteq H\cup B$ which is not
part of any loop reachable from~$s$. 

Without loss of generality,
we can
assume that there is no~$B$ edge which connects two vertices of the
path~$P^d_s$: we cannot update those edges anyway, and hence 
we can ignore 
them. If we simulate~$B$ edges with~$H$ edges, then 
the problem becomes equivalent to~$\SLF$ which is in~$P$. 
To see this,  suppose
$B=\{e_1,\ldots,e_k\}$, create a
new graph~$G'$ out of~$G$ by adding~$k$ vertices
$\{v^{e_1},\ldots,v^{e_k}\}$ to~$P^s_d$ to
obtain~$P^k_{s,d}=(s=v_1,v^{e_1},\ldots,v^{e_k},v_1,\ldots,v_n=d)$, and
  a set of edges~$H'=\{(u,v^{e_i})\mid e_i \in B, u = tail(e_i)\}$,
  where the tail of an edge~$e=(u,v)$ is~$u$. After that, 
  we delete all
  edges in~$B$. 
  We can now find the maximum set of the horizontal edges in
$G'$ which can be updated using the same 
algorithm as we had for~$\SLF$. If any edge~$H^e\in H'$ has
been chosen in the algorithm for~$\SLF$ in the~$G'$, we choose
$e\in E(G)$ for the update as well. These edges
together with all forward edges and the chosen edges from the set~$H$
in~$G'$ give us the maximum set of edges~$\mathcal H\in E(G)$ which
can safely be updated in the~$\RLF$ model in~$G$. 
Let~$\bar{\mathcal H} :=E(G) - \mathcal H$. 

Notice that there is no loop
reachable from~$s$ which uses only edges in~$H \cup F$ in
$G_{opt}=(V(G),\mathcal H)$, by the
construction of~$G'$. Moreover, there is no loop in~$G_{opt}$ which uses
edges in~$B$ and which is reachable from~$s$. To see the correctness of the
second claim, suppose an edge~$e=(u,v)\in B$ is chosen such that there is a
path~$P$ which goes through an edge~$e'\in H$ and connects~$s$ to
$u$. Then, in~$G'$ there was an edge~$e'' =(u,v^{e'})$. But at
least one of the edges~$e'$ and~$e''$ has been eliminated for 
the update in~$G'$
then, by the construction of the algorithm, either there is no edge 
like~$e$, or there is no path like~$P$ which goes through~$e'$. 

We
proved that the solution is valid. For the optimality,
we just note that there is a one-to-one relationship between simple loops of~$G$
which are reachable from~$s$, and loops of~$G'$. This means that if we make
$G'$ loop-free, we transfer~$G$ to the graph which has no loop
reachable from~$s$. So any optimal solution for~$G'$ is an optimal
solution for~$G$.
\end{proof} 
  
\noindent \textbf{Approximation Algorithms.}
Even in scenarios for which there is no optimal polynomial time 
scheduling algorithm, there can exist good approximations.
It is easy to observe that there is a reduction to the Maximum Acyclic Subgraph
Problem~(\emph{MASP}) which 
ensures that both $\RLF$ and $\SLF$ can be approximated at least as well as MASP. 
It is also easy to see that the problem for strong loop-freedom~(for~$\SLF$) is
$1/2$-approximable in general, as the problem boils down to 
finding a maximum subset of
$H$ edges which are safe to update, and 
at
least half of the~$H$ edges are pointing out to the left resp.~right,
and we can take the majority. 
Similarly for \RLF: let $F$ be the set of vertices where every
$v\in F$ appears along a walk between source and
destination. Similar to \SLF, at least half of the edges of $F$ are safe to
update, and we can find these edges quickly. Also every $e=(u,v)\in E(G)$, where
$u\not\in F$ or $v\not\in F$, is safe
to update. So we have at least a $1/2$-approximation.

However, for a small number of leaves, even 
better approximations are possible.
The following lemma can be proven by 
an approximation preserving reduction to the hitting set problem.
\begin{lemma}\label{lem:apx}
The optimal~$\SLF$ schedule 
is
$2/3$-approximable in polynomial time
in scenarios with exactly three leaves.  
For scenarios with four leaves, there exists a polynomial-time
$7/12$-approximation algorithm. 
\end{lemma}
\begin{proof}
We use an approximation preserving reduction to the~$d$-hitting set problem
which is~$\Sigma_{i=1}^d 1/i -
1/2$-approximable~\cite{Duh:1997:AKS:258533.258599}, 
and particularly, we
use a~$3$-hitting set which gives us a~$2/3$-approximation algorithm.

Let~$G=(V,E)$ be the update graph with at most three leaves 
and let~$H$ be the set of the horizontal edges. For every
closed simple loop~$C \subseteq G=(V,E)$ we have~$C^H = E(C) \cap H \neq \emptyset$. Furthermore~$C^H \le 3$. Given these
observations, we construct our hitting set as follows. Let~$|H| = m$ and
let~$F$ be a
one-to-one mapping~$F:H \rightarrow [m]$. For each simple
loop~$C_j$ let~$C_i^H$ =~$\{s_1,s_2,s_3\}$, and create a subset
 ~$S_i=\{F(s_1),F(s_2),F(s_3)\}$. Note that if~$|C_i^H| = 2$
  then we have a subset~$S_i$ of size~$2$. There are at most~$|H|^3$
  simple loops, as choosing any set of size at most~three edges from~$E$
  forces at most one simple loop. So we have $m \choose 3$
  loops with~$3$ edges in~$H$ and~$m \choose 2$ loops with two edges
  in~$H$. Furthermore the hitting set for
 ~$S_1,\ldots,S_t$ gives a minimum set of update edges to be removed;
  on the other hand, every subset~$S_i$ is of cardinality at most
 ~$3$. This gives a~$4/3$-approximation on the size of subset
 ~$H'\subseteq H$, which we do not update. On the other hand, in the
  optimal solution~$H^{opt}$ we have~$|H'| \le H$ resp.~$|H^{opt}| \ge |H'|$, so the approximation factor will be at least
 ~$(1-1/3)|opt|$: this is a~$2/3$-approximation, as claimed. For four
  leaves, a similar argument works, and we omit the proof.
\hfill\qed
\end{proof}


 

\section{Related Work}\label{sec:relwork}


In their seminal work, Reitblatt et
al.~\cite{abstractions} initiated the study of
network updates providing strong, per-packet consistency guarantees, and the authors also presented
a 2-phase commit protocol. This protocol also forms the basis of the distributed control plane implementation 
in~\cite{infocom15}.
Mahajan and Wattenhofer~\cite{roger} started investigating
a hierarchy of transient consistency properties---in particular also~(strong) loop-freedom
but for example also bandwidth-aware updates~\cite{Brandt2016On}---for destination-based routing
policies.
The measurement studies in~\cite{dionysus} and~\cite{knowFlow} provide
empirical evidence for the non-negligible time and high variance of switch updates, further motivating their and our work.
In their paper, Mahajan and Wattenhofer proposed an algorithm to ``greedily'' select a maximum number of edges which can be used early 
during
the policy installation process. This study was recently refined 
in~\cite{Forster2016Consistent,Forster2016Power},
a parallel work to ours, where the authors also establish a hardness result
for destination based routing (single- and multi-destination).
Our work builds upon~\cite{roger} and 
complements the results in~\cite{Forster2016Consistent,Forster2016Power}:
We consider  
the scheduling complexity of updating \emph{arbitrary routes}
which are not necessarily destination-based. Interestingly,
our results (using a different reduction) show that
even with the requirement that the initial and the final routes
are simple paths, the problem is NP-hard. 
Moreover, our results hold for both the strong $\SLF$
and the relaxed $\RLF$ loop-free problem
variants introduced in~\cite{podc15} (this distinction
does not exist in~\cite{Forster2016Consistent}). 
The $\SLF$ can be seen as a special
variant of the Dual Feedback Arc Set Problem~(FASP) resp.~Maximum
Acyclic Subgraph Problem~(MASP): important classic problems
in approximation theory~\cite{fasp}. In particular,
it is known that dual-FASP/MASP can be $1/2+\varepsilon$ 
approximated on general graphs~(for arbitrary small $\varepsilon$). 
The results presented in this
paper also imply that better approximation algorithms and even optimal
polynomial-time algorithms exist for special graph families, namely
graph families describing network update problems; this may be of
independent interest.
The~$\RLF$ variant is a new optimization problem,
and to the best of our knowledge, existing bounds are not applicable
to this problem. 
We should note that FASP is in FPT~\cite{Chen:2008:FAD:1411509.1411511}, and
the hitting set problem is W[2]-hard~\cite{DBLP:series/txtcs/FlumG06}. In
our hardness construction we actually find a reduction from hitting
set to FASP for particular graph classes. But the reduction is not
parameter preserving, so the W-hierarchy does not collapse.
Finally, our model is orthogonal to the network update problems
aiming to minimizing the number of interactions with the controller
(the so-called \emph{rounds}), which we have recently studied
for single~\cite{podc15} and multiple~\cite{dsn16} policies, also including additional properties,
beyond loop-freedom, such as waypointing~\cite{sigmetrics16}. The two
objectives conflict~\cite{podc15}, a good approximation for the number of update
edges yields a bad approximation for the number of rounds, and vice versa.

\section{Concluding Remarks: Special Graph Classes}\label{sec:special}

We conclude our contribution with some remarks.
First, it is easy to observe that there is a reduction to the Maximum Acyclic Subgraph
Problem~(\emph{MASP}) which 
ensures that both $\RLF$ and $\SLF$ can be approximated at least as well as MASP. 

It is also interesting to study the hardness of the problem on some special graph
classes. By $\bar{G}$ we denote the underlying undirected graph of a
graph $G$ which is obtained by replacing directed edges with undirected
edges, and deleting parallel edges.
We have the following lemma for bounded $\tw$~\cite{Robertson1986309}
scenarios. 

\begin{lemma}\label{lem:btw}
Given a digraph $G$, both $\RLF$ and $\SLF$ are solvable in polynomial time if
$\bar{G}$ has bounded \tw.
\end{lemma}
\begin{proof}
Thanks to Courcelle's theorem~\cite{COURCELLE199012},
we can solve the feedback
arc set problem in bounded $\tw$ digraphs in polynomial time. This
directly gives a solution for 
SLF. For \RLF, we find all vertices which are on
a walk between source and destination, and we apply Courcelle theorem
to the subgraph of $G$ induced by these vertices. 
\end{proof}

Unfortunately the undirected width measures are not very useful in
directed graphs. Analogously to $\tw$ which is defined for 
undirected graphs, there
exists a directed $\tw$ notion for directed graphs, introduced by Johnson et
al.~\cite{Johnson2001138}. We refer the reader to the provided
reference for the definition.

An interesting question regards whether 
$\RLF$ and \SLF, and more
generally MASP, are polynomial-time solvable in digraphs of bounded
directed treewidth and bounded degree. There are two negative
results related to this question. First, it has been shown that the
Feedback Arc Set Problem~(FASP) is already 
NP-complete~\cite{Kreutzer20114688} in digraphs of
directed tree width at most~$5$. Their hardness construction
is based on a graph which has a bounded degree in all vertices except
for one vertex. It seems that with binarization one can easily adapt their
proof to show that the FASP problem still remains hard in digraphs of
bounded degree and bounded directed treewidth. But on bounded degree
graphs, vertex cover problems~\cite{Garey:1979:CIG:578533} are NP-complete,
and a simple construction for vertex cover yields an NP-hardness
result for FASP in those graphs as well. 
This suggests that directed tree-width cannot be exploited in our problem. 

However, another kind of directed width measure may be more useful: 
The \emph{directed path width} is defined very similarly
to the directed
treewidth; intutively the graph looks like a ``thick directed path''. None
of the negative results for bounded degree graphs on graphs of bounded
directed tree-width can be extended to digraphs of bounded directed
path-width with bounded degree. 
We claim the following: 
There is a function $f\colon \N \rightarrow \N$ such that for a digraph $G$ of directed path-width $k$ and maximum degree
$d$, there is an algorithm which runs in time and space $n^{f(k+d)}$
and finds an optimal solution to FASP.

\noindent \textbf{Acknowledgments.}
The research of Saeed Amiri has been supported by the
European Research Council (ERC) under the European Union's Horizon
2020 research and innovation programme (ERC consolidator grant DISTRUCT,
agreement No 648527).

{

}

\end{document}